\documentclass[12pt]{article}

\usepackage{amsmath,amsfonts,amsthm,amscd,amssymb,slashed}
\usepackage[pdftex]{graphicx}
\usepackage{algpseudocode}
\usepackage{tikz}
\usepackage{enumerate}
\usepackage[hidelinks]{hyperref}
\usepackage{mathtools}
\usetikzlibrary{shapes,arrows}
\usetikzlibrary{matrix,arrows}
\usepackage{epstopdf}
\usepackage{xcolor,import}
\usepackage{transparent}
\usepackage[all]{xy}

\usepackage{geometry}
\makeatletter

\@addtoreset{equation}{section}
\def\section{\@startsection {section}{1}{\z@}{-2.5ex plus -1ex minus
 -.2ex}{1.3ex plus .2ex}{\large\bf}}
\def\subsection{\@startsection{subsection}{2}{\z@}{-2.25ex plus%
 -1ex minus -.2ex}{0.5ex plus .2ex}{\bf}}

\newcommand{\R}{\mathbb{R}}

\newcommand{\C}{\mathbb{C}}
\newcommand{\Z}{\mathbb{Z}}
\newcommand{\CP}{\mathbb{C}P}

\newtheorem{thm}{Theorem}[section]
\newtheorem{prop}[thm]{Proposition}
\newtheorem{cor}[thm]{Corollary} 
\newtheorem{dfn}[thm]{Definition}
\newtheorem{lem}[thm]{Lemma}
\theoremstyle{definition}
\newtheorem{ex}[thm]{Example}

\numberwithin{equation}{section}
\newcommand{\ud}{\mathrm{d}}
\newcommand{\ui}{\mathrm{i}}
\newcommand{\NW}{\mathcal{N}}

\begin{document}
\parskip 6pt
\parindent 0pt

\vspace{.4cm}

\begin{center}
{\Large \bf Vortex Harmonic Spinors on the Nappi--Witten Space}

\vspace{0.4 cm}

{\bf Calum Ross}\\
{Department of Computer Science, Edge Hill University, Ormskirk L39 4QP, United Kingdom\\
Research and Education Center for Natural Sciences, Keio
University, Hiyoshi 4-1-1, Yokohama, Kanagawa 223-8521, Japan\\
calum.ross@edgehill.ac.uk}\\
\vspace{0.5cm}
{\bf Ra\'{u}l S\'{a}nchez Gal\'{a}n}\\
 4i Intelligent Insights, Tecnoincubadora Marie Curie, PCT Cartuja,
41092 Sevilla, Spain,\\
{r.sanchez@4i.ai }
\end{center}
\abstract{
We establish a correspondence between vortex equations on flat Riemann surfaces and harmonic spinors on the Nappi--Witten space, the group manifold of a central extension of the Euclidean group $SE(2)$. Vortex configurations lift naturally to this setting, producing explicit solutions of a twisted Dirac equation. Using the conformal flatness of the Nappi--Witten metric, these solutions induce harmonic spinors on four-dimensional Minkowski space. This yields a geometric construction of Abelian magnetic zero-modes on flat Minkowski spacetime from vortex data.
}

\tableofcontents

\section{Introduction}
Vortices are examples of solitons which occur in the Abelian-Higgs model and its variants \cite{MS}. They are the minimisers of Abelian-Higgs type energy functionals in a given topological sector, and their topological charge is the winding number, which is equal to the number of zeros, counted with multiplicity, of the scalar Higgs field. In \cite{Manton:2016waw}, a unifying picture is given of five different types of vortices in terms of generalisations of the Abelian-Higgs model. In a series of papers \cite{Ross2021,RS2,RS1}, all of these vortices were related to the geometry of three-dimensional Lie groups. There was also a link established between Abelian vortices and magnetic zero-modes, these are spinors which are harmonic with respect to a Dirac operator which has been twisted by an Abelian gauge field. The Lie groups that are related to vortices are: $SU(2)$ the special unitary group, $SU(1,1)$ the special pseudo unitary group, and $SE(2)$ the Euclidean group in two dimensions. Most of the vortices from \cite{Manton:2016waw} are related to $SU(1,1)$, the hyperbolic, Bradlow, and Ambj{\o}rn-Olesen vortices, while Popov vortices are related to $SU(2)$. In both cases, a bi-invariant metric on the group manifold and an associated Dirac operator can be constructed. However, for Jackiw--Pi vortices and Laplace vortices\footnote{Laplace vortices were not considered in \cite{Manton:2016waw} as they correspond to the pair of a flat connection and a covariantly holomorphic section and there are no nontrivial solutions with finite energy. However, in \cite{Contatto:2017alh} it was pointed out that the Laplace vortex equations are still perfectly sensible to consider.}, the group manifold is $SE(2)$, which has a degenerate Killing form and thus a degenerate metric, meaning that we cannot consider harmonic spinors in the same way.\\

A resolution comes in the form of centrally extending the Lie group. It is well known that centrally extending $SE(2)$ leads to the Nappi--Witten or diamond Lie group \cite{II,NW}, which admits a family of invariant Lorentzian metrics. The lifting construction for vortices used in \cite{Ross2021} extends to the central extension, and thus there are still harmonic spinors corresponding to Jackiw--Pi and Laplace vortices. In this paper we carry out this generalisation and demonstrate what vortex harmonic spinors on Nappi--Witten look like. We then make use of the fact that the Nappi--Witten space is conformally flat to construct vortex harmonic spinors on Minkowski space.\\

The paper is organised as follows. In Section~\ref{sec: NW and Dirac} we introduce our conventions, the Nappi--Witten space, and construct a Dirac operator. In Section~\ref{sec: vortices} we review the theory of Jackiw--Pi and Laplace vortices and show how to lift them to vortex configurations on the Nappi--Witten space. In Section~\ref{sec: vortex harmonic spinors on NW} we provide the details of our results on vortex harmonic spinors on the Nappi--Witten space. Section~\ref{sec: vortex harmonic spinors on Minkowski} reviews the conformal relationship between the Nappi--Witten space and four-dimensional Minkowski space and shows how vortex harmonic spinors on Nappi--Witten give rise to harmonic spinors on Minkowski. 
Finally, Section~\ref{Sec: summary} summarises our results and discusses some directions for further work.
\section{The Nappi--Witten space and its Dirac operators}
\label{sec: NW and Dirac}
\subsection{The Nappi--Witten space} \label{ssec:Group_conventions}
The Nappi--Witten space, $\NW$, is a four-dimensional Lie group first introduced in \cite{NW}. It is also known as the diamond Lie group and is an example of a plane wave spacetime. Its Lie algebra is a solvable four-dimensional Lie algebra obtained as a central extension of the Lie algebra of $SE(2)$. We work with the generators $P_{1},P_{2},J,T$, which satisfy the commutation relations
\begin{equation}
[J,P_{i}]=\epsilon_{ij}P_{j}, \quad [P_{i},P_{j}]=\epsilon_{ij}T,
\end{equation}
with all the other commutators vanishing. From a physics perspective, the $P_{i}$ generate translations in the coordinates $x, y$, the element $J$ generates rotations in the plane spanned by $ x$ and $y$, and $T$ is a central element. Since the Nappi--Witten Lie algebra is solvable, its Killing form is degenerate. However, it has the following two-parameter family of non-degenerate
invariant symmetric bilinear forms \cite{Cangemi_1992,NW}, whose Gram
matrix in the basis \((P_1,P_2,J,T)\) is
\begin{equation} 
    \langle \cdot , \cdot \rangle_{k,b}
    =
    k\begin{pmatrix}
        1&0&0&0\\
        0&1&0&0\\
        0&0&b&1\\
        0&0&1&0
    \end{pmatrix}.
    \label{eq: NW quadratic form}
\end{equation}
Here we take \(k=1\) and denote this bilinear form simply by
\(\langle \cdot , \cdot \rangle_b\). It is possible to set $b=0$ using a suitable change of basis, see equation \eqref{eq:t_translate}. We keep $b$ in some of the formulas below for generality, and specialise to $b=0$ when needed.

The Nappi--Witten space is diffeomorphic to a cylinder $\R^{3}\times S^{1}$, and any element $h \in \NW$ can be written using the exponential map as
\begin{equation}
    h(x,y,\theta,t)=e^{xP_{1}+yP_{2}}e^{\theta J +t T},
    \label{eq: Nw element}
\end{equation}
where $\theta\in S^{1}$ and $x,y,t\in \R$. Note that in some references, a different choice is made for how to write elements of $\NW$ in terms of the exponential map. However, these are all related by multiplication by an element of $\NW$.\\

The Maurer--Cartan form of the Nappi--Witten space is given by
\begin{equation}
\begin{split}
    h^{-1}dh =&\left(\cos\theta\ud x +\sin\theta \ud y\right)P_{1}+\left(\cos\theta \ud y -\sin\theta \ud x\right)P_{2}+\ud \theta J\\
    &+\left(\ud t -\frac{1}{2}x\ud y +\frac{1}{2}y\ud x\right)T.
    \end{split}
\end{equation}
From this, we can read off the left-invariant one-forms as
\begin{equation}
\label{eq:lf one forms}
    \begin{split}
        \sigma^{1}&= \cos\theta\ud x +\sin\theta \ud y,\\
        \sigma^{2}&= \cos\theta \ud y -\sin\theta \ud x,\\
        \sigma^{3}&=\ud \theta,\\
        \sigma^{4}&= \ud t -\frac{1}{2}x\ud y +\frac{1}{2}y\ud x.
    \end{split}
\end{equation}

The above family of Ad-invariant symmetric bilinear forms on the Lie algebra give rise to bi-invariant metrics on the Nappi--Witten space with signature $(-,+,+,+)$:
\begin{equation}
    ds^2 = \langle h^{-1}dh, h^{-1}dh \rangle_b.
\end{equation}

In the above left-invariant one-form basis, one obtains
\begin{equation}
    ds^2 = (\sigma^1)^2 + (\sigma^2)^2 + b (\sigma^3)^2 + 2 \,\sigma^3 \sigma^4,
\end{equation}
and hence in the $(x,y,\theta,t)$-coordinates,
\begin{equation}
    ds^{2}=\ud x^{2}+\ud y^{2}+2\ud \theta \ud t +(y\ud x-x\ud y)\ud \theta +b\ud \theta^{2}.
    \label{eq: NW metric}
\end{equation}
These metrics were first derived in \cite{NW} and then discussed in some detail in \cite{FS}. A discussion of central extensions of this kind and their relevance to physics is given in \cite{PS}.

Raising the index of the left-invariant one-forms \(\sigma^i\) with respect
to the metric \(ds^2\), we obtain the associated left-invariant vector fields
\(Y_i=(\sigma^i)^\sharp\):
\begin{equation}
    \begin{split}
        Y_{1}&=\cos\theta\partial_{x}+\sin\theta\partial_{y}+\frac{1}{2}\left(-y\cos\theta+x\sin\theta\right)\partial_{t},\\
        Y_{2}&= -\sin\theta\partial_{x}+\cos\theta\partial_{y}+\frac{1}{2}\left(x\cos\theta+y\sin\theta\right)\partial_{t},\\
        Y_{3}&=\partial_{t},\\
        Y_{4}&= \partial_{\theta}-b\partial_{t}.
    \end{split}
    \label{eq: left-invariant vfs}
\end{equation}
The commutation relations between the $Y_{i}$ match those of the Lie algebra of the Nappi--Witten Lie group. \\

To obtain an orthonormal frame via the Gram-Schmidt procedure, we take the frame 
\begin{equation}
    Y_{1}, \quad Y_{2}, \quad Y_{3}-Y_{4}, \quad Y_{3}+Y_{4},
\end{equation}
which contains no lightlike vector field ($\langle Y_{3},Y_{3}\rangle=0$), and use the Gram-Schmidt procedure to turn this frame into the orthonormal frame
\begin{equation}
    \begin{split}
    U_{0}&=\frac{1}{\sqrt{b+2}}\left(-\partial_{\theta}+\left(b+1\right)\partial_{t}\right),\\
        U_{1}&=Y_{1},\\
        U_{2}&=Y_{2},\\
        U_{3}&=\frac{1}{\sqrt{b+2}}\left(\partial_{\theta}+\partial_{t}\right).
    \end{split}
    \label{eq:orthonormal basis of vectors}
\end{equation}
The commutation relations in this frame are
\begin{align}\label{Comm_U}
    [U_{0},U_{i}]&=-\frac{\epsilon_{ij}}{\sqrt{b+2}}U_{j}, \quad i,j =1,2 \notag\\
    [U_{3},U_{i}]&=\frac{\epsilon_{ij}}{\sqrt{b+2}}U_{j}, \quad i,j =1,2 \\
    [U_{1},U_{2}]&=\frac{1}{\sqrt{b+2}}\left(U_{0}+U_{3}\right).\notag
\end{align}
The non-zero pairings of the above left-invariant one-forms and this  orthonormal basis are
\begin{equation}
\begin{split}
    \sigma^{1}(U_{1})&=1, \qquad \sigma^{2}(U_{2})=1,\\ 
    \sigma^{3}(U_{0})&=-\frac{1}{\sqrt{b+2}}, \quad \sigma^{3}(U_{3})=\frac{1}{\sqrt{b+2}},\\
     \sigma^{4}(U_{0})&=\frac{b+1}{\sqrt{b+2}},\quad \sigma^{4}(U_{3})=\frac{1}{\sqrt{b+2}}.
\end{split}
\end{equation}

The coframe dual to the orthonormal frame
\((U_0,U_1,U_2,U_3)\), satisfying $U^\mu(U_\nu)=\delta^\mu_{\nu}$ is given by
\begin{equation}
\begin{aligned}
U^0 &= \frac{-\sigma^3+\sigma^4}{\sqrt{b+2}}, 
&\qquad U^1 &= \sigma^1,\\
U^2 &= \sigma^2,
&\qquad U^3 &= \frac{(b+1)\sigma^3 + \sigma^4}{\sqrt{b+2}},
\end{aligned}
\end{equation}
and the metric takes the diagonal form
\begin{equation}
ds^2 = -(U^0)^2+(U^1)^2+(U^2)^2+(U^3)^2.
\end{equation}
\begin{prop} In terms of the orthonormal coframe $\{U^{\mu}\}$, the Levi-Civita connection matrix of one-forms $\omega^{\mu}_{\;\;\nu}$ has the form
    \begin{equation}
       \left(\omega^{\mu}_{\;\;\nu} \right)=\frac{1}{2\sqrt{b+2}}\begin{pmatrix}
           0&-U^{2}&U^{1}&0\\
           -U^{2}&0&U^{0}-U^{3}&U^{2}\\
           U^{1}&-U^{0}+U^{3}&0&-U^{1}\\
           0&-U^{2}&U^{1}&0
       \end{pmatrix}.
       \label{eq:NW levi-civita connection matrix}
    \end{equation}
\end{prop}
\begin{proof}
    The connection matrix $\omega^{\mu}_{\;\;\nu}$ of the Levi-Civita connection $\nabla$ is valued in $\mathfrak{so}(1,3)$ and in the frame $\{U_{\mu}\}$ is given by
    \begin{equation}
        \nabla U_{\nu}=\sum_{\mu}\omega^{\mu}_{\;\;\nu}\otimes U_{\mu}.
    \end{equation}

    Since the frame is left-invariant and the Levi-Civita connection comes from a bi-invariant metric, the covariant derivatives can be written as $\nabla_{U_\mu}U_\nu=\frac{1}{2}[U_\mu,U_\nu]$. Using the commutation relations \eqref{Comm_U}, one can directly compute the components to find the expression in equation \eqref{eq:NW levi-civita connection matrix}.
\end{proof}

\subsection{The Dirac operator on Nappi--Witten space}
As every Lie group is parallelizable, the tangent bundle of $\mathcal{N}$ is trivial. Consequently, all of the Stiefel--Whitney classes vanish, in particular, the second class vanishes, and $\mathcal{N}$ admits spin structures \cite{lawson2016spin}. The set of spin structures
forms an affine space over $H^1(\mathcal{N},\mathbb{Z}_2)$. Using the universal coefficient theorem and Hurewicz's theorem, one can show that $H^1(\mathcal{N},\mathbb{Z}_2)\cong \mathrm{Hom}\bigl(\pi_1(\mathcal{N}),\mathbb{Z}_2\bigr)$ and, since $\pi_1(\mathcal{N}) = \Z$, we conclude that there are two inequivalent spin structures on $\mathcal{N}$.
We work primarily with the trivial spin structure, for which the spinor bundle is globally trivial and spinors are globally defined smooth functions $\Psi:\mathcal{N}\to\mathbb{C}^4$. The nontrivial spin structure may be obtained by twisting the trivial bundle with a flat complex line bundle $L_{-}$ over $S^1$ whose holonomy is $-1$. Equivalently, this may be realised by a flat $U(1)$-connection represented locally by $\frac12\,d\theta$. Spinors in this nontrivial spin structure satisfy
$\Psi(\theta+2\pi)=-\Psi(\theta)$.

The Dirac operator acts on spinors $\Psi:\mathcal{N}\to\mathbb{C}^4$
as
\begin{equation}
D \Psi=   \sum_{\mu=0}^3 \gamma^\mu \nabla_{U_\mu} \Psi,
\end{equation}
where \(\nabla \) is the spinor covariant derivative, which is built from the Levi-Civita connection via the spin representation,
 $\{U_{\mu}\}$ is a global orthonormal frame and the gamma matrices satisfy $\{\gamma^{\mu},\gamma^{\nu}\}=-2g^{\mu\nu} I_4$, with $\left(g^{\mu\nu}\right)=\text{diag}\left(-1,1,1,1\right)$. We write a spinor as 
\begin{equation}
\Psi = \sum_{i=1}^{4} f_{i}\,\Psi_{i},
\end{equation}
where $f_{i}:\mathcal{N}\to\mathbb{C}$ are smooth functions and 
$\{\Psi_{i}\}$ is the standard basis of $\mathbb{C}^{4}$. The Dirac operator is 
 \begin{equation}    D\left(\sum_{i}f_{i}\Psi_{i}\right)=\sum_{i}\left(\sum_{\mu}\gamma^{\mu} U_{\mu}(f_{i})\Psi_{i} + f_{i}D\Psi_{i}\right).
    \label{eq:NW_dirac_operator}
\end{equation}

We take the four gamma matrices to be in the Weyl  representation,
\begin{equation}
\begin{split}
    \gamma^{0}&=\begin{pmatrix}
        0&0&1&0\\
        0&0&0&1\\
        1&0&0&0\\
        0&1&0&0
    \end{pmatrix}, \qquad \gamma^{1}=\begin{pmatrix}
        0&0&0&1\\
        0&0&1&0\\
        0&-1&0&0\\
        -1&0&0&0
    \end{pmatrix}\\
    \gamma^{2}&=\begin{pmatrix}
        0&0&0&-\ui\\
        0&0&\ui&0\\
        0&\ui&0&0\\
        -\ui&0&0&0
    \end{pmatrix}, \qquad \gamma^{3}=\begin{pmatrix}
        0&0&1&0\\
        0&0&0&-1\\
        -1&0&0&0\\
        0&1&0&0
    \end{pmatrix},
    \end{split}
\end{equation}
 and use the global orthonormal frame of left-invariant
vector fields defined in \eqref{eq:orthonormal basis of vectors}. With these conventions, 
\begin{equation}
\sum_{\mu}\gamma^{\mu}U_{\mu}(f_{i})=
\begin{pmatrix}
0 & 0 & U_0 + U_3 & U_1 - \ui U_2 \\
0 & 0 & U_1 + \ui U_2 & U_0 - U_3 \\
U_0 - U_3 & -\,U_1 + \ui U_2 & 0 & 0 \\
-\,U_1 - \ui U_2 & U_0 + U_3 & 0 & 0
\end{pmatrix}\!(f_i)
\end{equation}
and the action of the Dirac operator on the basis spinors is
\begin{equation}
D \Psi_i= \sum_\mu \gamma^\mu \nabla_\mu \Psi_i= \frac{1}{4}\sum_{\mu \nu \rho} \gamma^\mu \omega_{\nu\rho} (U_\mu)\gamma^{\nu}\gamma^{\rho}\Psi_i.
\end{equation}
Here, the lowered connection one-forms are given by $\omega_{\nu\rho} =\sum_\kappa \omega^\kappa_{\,\rho} g_{\kappa \nu}$ and satisfy the skew-symmetry property $\omega_{\nu\rho} = - \omega_{\rho \nu}$. Because each entry $\omega_{\nu\rho}$ is itself a one-form, the spin connection coefficients $\omega_{\nu\rho}(U_\mu)$ are obtained by evaluating these one-forms on the frame vectors. A direct computation then yields
     \begin{align}
     D\Psi_{i} &= \frac{\ui}{2\sqrt{b+2}} \begin{pmatrix} 0 & 0 & 0 & 0 \\ 0 & 0 & 0 & 3 \\ -3 & 0 & 0 & 0 \\ 0 & 0 & 0 & 0 \end{pmatrix} \Psi_{i}.
\end{align}

\emph{Harmonic spinors} are spinors $\Psi$ such that $D\Psi =0$; they are sometimes also called \emph{zero-modes}. 
In our case, a spinor $\Psi = \sum_{i=1}^{4} f_{i}\,\Psi_{i}$ is harmonic if
\begin{equation}
\begin{aligned}
    (\sqrt{2}\,\partial_t) f_3 + 2 e^{i\theta}\!\left(\partial_z - \tfrac{i}{4}\bar{z}\,\partial_t\right) f_4 &= 0 \\[10pt]
    2 e^{-i\theta}\!\left(\partial_{\bar{z}} + \tfrac{i}{4}z\,\partial_t\right) f_3 + \left(-\sqrt{2}\,\partial_\theta + \tfrac{3i}{2\sqrt{2}}\right) f_4 &= 0 \\
    \left(-\sqrt{2}\,\partial_\theta - \tfrac{3i}{2\sqrt{2}}\right) f_1 + 2 e^{i\theta}\!\left(-\partial_z + \tfrac{i}{4}\bar{z}\,\partial_t\right) f_2 &= 0 \\
    -2 e^{-i\theta}\!\left(\partial_{\bar{z}} + \tfrac{i}{4}z\,\partial_t\right) f_1 + (\sqrt{2}\,\partial_t) f_2 &= 0
\end{aligned}
\end{equation}
where  
$z := x + i y$, 
$\bar{z} := x - i y$,  
 $\partial_{z} = \tfrac{1}{2}\,(\partial_{x} - i\,\partial_{y})$,  
$\partial_{\bar{z}} = \tfrac{1}{2}\,(\partial_{x} + i\,\partial_{y})$.
One can easily find solutions that depend only on $\theta$. In this case, the system reduces to two uncoupled first-order ODEs for $f_1$ and $f_4$, while $f_2$ and $f_3$ are left completely unconstrained. One easily finds that the local general solution is
\begin{equation}
\Psi(\theta)=
\begin{pmatrix}
C_1 e^{-3i\theta/4}\\
g_2(\theta)\\
g_3(\theta)\\
C_4 e^{3i\theta/4}
\end{pmatrix},
\end{equation}
where $C_1,C_4\in\mathbb C$ and $g_2,g_3$ are arbitrary smooth functions of $\theta$. Therefore, for the trivial spin structure, the global $\theta$-dependent harmonic spinors are obtained by setting $C_1=C_4=0$ and requiring $g_2$ and $g_3$ to be $2\pi$-periodic, while for the nontrivial spin structure they are obtained by setting $C_1=C_4=0$ and requiring $g_2$ and $g_3$ to be $2\pi$-anti-periodic.

If the spin bundle is twisted by a line bundle equipped with the Abelian connection $A$ then the twisted Dirac operator $D_{A}$ has the same form as in equation \eqref{eq:NW_dirac_operator} but with $U_{\mu}$ replaced by $U_{\mu}+iA_{\mu}$.
Since twisting by a line bundle is equivalent to coupling to a magnetic field, spinors that are harmonic with respect to a twisted Dirac operator are sometimes called \emph{magnetic zero-modes}. The case of magnetic zero-modes and their applications to physics has been well studied with a primary focus on $\R^{3}$ starting in \cite{LY} and in subsequent work including \cite{Adam:1999mq,Dunne:2008qn,Erdos:2001,Min:2009cc}.
\section{Vortices}
\label{sec: vortices}
\subsection{Geometric conventions}\label{sec: Geometric conventions}
We work on the Riemann surface
$M_{0}=\C$ which has metric and K{\"a}hler form
\begin{align}
    ds^{2}_{M_0}&=4\left[\left(\ud x^{1}\right)^{2}+\left(\ud x^{2}\right)^{2}\right]=(e^{1})^{2}+(e^{2})^{2},\\
    \omega&=e^{1}\wedge e^{2}=4\ud x^{1}\wedge\ud x^{2}.
\end{align}
For us, $M_{0}$ is always assumed to be $\C$, since JP vortices on a flat torus correspond to vortices on $\C$, invariant under the action of a discrete subgroup, the lattice defining the torus. The choice of the overall factors of $4$ in both the metric and the volume form is so that the conventions here agree with those in \cite{Contatto:2017alh,Ross2021,RS1} for the general case.
We write these in terms of a complex coframe $\{e, \bar{e} \}$ with 
\begin{equation}
    e=2\ud z=e^{1}+\ui e^{2}.
\end{equation}

In \cite{Ross2021} and related work, vortices were lifted to the Lie groups $SU(2),\, SU(1,1)$, and $SE(2)$ by using the fact that all three groups are circle bundles over Riemann surfaces, $\CP^{1}, H^{2}$, and $\C$ respectively. Here we view $\NW$ as a trivial bundle over $\C$ with fibre $S^{1}\times \R$, the bundle projection is the semi-Riemannian submersion
\begin{equation}
    \pi: \NW \to M_{0}, \qquad (x,y,\theta,t)\mapsto \left(\frac{y}{2},-\frac{x}{2}\right)=(x^{1},x^{2}).
\end{equation}
The section 
\begin{equation}
    s:(x^{1},x^{2})\mapsto (-2x^{2},2x^{1},0,0)=(x,y,\theta,t),
\end{equation}
gives an isometric embedding of $(M_0,ds^2_{M_0})$ into $(\mathcal{N},ds^2)$.  
At the level of the group element $h\in \NW$ this projection is $\frac{1}{2}$ of the $P_1$, $P_2$ components of $hJh^{-1}$, and is an analogue of the Hopf projection. 

The next proposition gives the explicit relation between the  coframe $e^{1},e^{2}$ on $M_{0}$ and the left-invariant one-forms $\sigma^1, \sigma^2$ on $\NW$.

\begin{prop}
The pullback of the coframe on $M_{0}$ via $\pi:\NW\to M_{0}$ yields, 
\begin{equation}\label{eq:pullbackcoframe}
\pi^{*}\begin{pmatrix}
e^{1}\\
e^{2}
\end{pmatrix} =R\left(-\frac{\pi}{2}\right) R\left(\theta\right)\begin{pmatrix}
\sigma^{1}\\
\sigma^{2}
\end{pmatrix},
\end{equation}
where $R(\theta)$ is a rotation by $\theta$.
\end{prop} 
The proof of this is a direct computation.
 \begin{proof}
 Note that
\begin{align}
\pi^{*}\left(e^{1}\right)&= 2\pi^{*}\left(\ud x^{1}\right)=\ud y,\\
\pi^{*}\left(e^{2}\right)&= 2\pi^{*}\left(\ud x^{2}\right)=-\ud x,
\end{align}
then 
\begin{align}
 R\left(-\theta\right)R\left(\frac{\pi}{2}\right)\pi^{*}\begin{pmatrix}
e^{1}\\
e^{2}
\end{pmatrix}&=\begin{pmatrix}
\cos\theta&\sin\theta\\
-\sin\theta&\cos\theta
\end{pmatrix}\begin{pmatrix}
0&-1\\
1&0
\end{pmatrix}\begin{pmatrix}
\ud y\\
-\ud x
\end{pmatrix}\\
&=\begin{pmatrix}
\cos\theta&\sin\theta\\
-\sin\theta&\cos\theta
\end{pmatrix}\begin{pmatrix}
\ud x\\
\ud y
\end{pmatrix}\\
&=\begin{pmatrix}
\cos\theta \ud x +\sin\theta \ud y\\
\cos\theta \ud y-\sin\theta \ud x
\end{pmatrix}\\
&=\begin{pmatrix}
\sigma^{1}\\
\sigma^{2}
\end{pmatrix}.
\end{align}
 Multiplying on the left by the inverse matrix $R\left(-\frac{\pi}{2}\right)R\left(\theta\right)$ gives the desired result.
\end{proof}

Note that $s^{*}\sigma^{4}=2\left(x^{1} \ud x^{2}-x^{2}\ud x^{1}\right) $ and hence $\ud\left(s^{*}\sigma^{4}\right)=e^{1}\wedge e^{2}$ is the area form on $M_{0}$.

\subsection{Jackiw--Pi and Laplace vortices}
We take the same definition of a vortex as in \cite{Ross2021}.
\begin{dfn}
    A vortex on a Riemann surface $(M_{0},g_0)$ is a pair $(a,\phi)$ consisting of a connection $a$ on a $U(1)$-bundle over $M_{0}$ and a smooth section $\phi$ of the associated line bundle (via the standard representation of $U(1)$ on $\C$)  which satisfy the vortex equations
\begin{equation}
    \bar{\partial}_a \phi = 0, \qquad F=\ud a =\left(\lambda_{0}-\lambda\vert \phi\vert^{2}\right) \ud\mathrm{vol}_{g_0}. \label{eq: Vortex equations}
\end{equation}
\end{dfn}

The first Chern number of the $U(1)$-bundle gives the winding number, or topological charge of the vortex. In \cite{Manton:2016waw} it was shown that these equations are integrable when $\lambda_{0}=-K_{0},\lambda=-K$ are the constant Gauss curvature of two metrics on $M_0$. In the integrable case, the vortex equations reduce to the Liouville equation and are solved by a holomorphic map $f:M_{0}\to M$ between Riemann surfaces with Gauss curvatures $K_{0},K$. \\

Here we are interested in the case where $M_{0}$ is flat so $\lambda_{0}=0$. There are then two possibilities for $\lambda$: $\lambda=0$ Laplace vortices, and $\lambda=-1$ Jackiw--Pi (JP) vortices. From a physics point of view, equation \eqref{eq: Vortex equations} arises from minimising an energy functional and, when $M_{0}$ is non-compact, it needs to be supplemented by the condition $\vert \phi \vert \to 1$ to ensure the energy is finite. This finite energy condition means that solutions of the Laplace vortex equations that satisfy the boundary condition are trivial up to gauge. However, since the torus is a flat compact Riemann surface, Laplace vortices on the torus can be non-trivial.\\

In local complex coordinates $z=x^{1}+\ui x^{2}$ the first vortex equation is,
\begin{equation}
    (\partial_{\bar{z}} - \ui a_{\bar{z}}) \phi =0,
\end{equation}
which says that the Higgs field $\phi$ is holomorphic with respect to the connection $a$. Note that for JP vortices the second vortex equation becomes
\begin{equation}
    F = \ud a =\frac{\ui}{2}\vert\phi\vert^{2}e\wedge \bar{e}.
\end{equation}

JP vortices on the complex plane $\C$ have been well studied due to their relationship to Chern-Simons theory in $2+1$ dimensions. For the details of this relationship, we refer the reader to \cite{Horvathy:2008hd}. A discussion of how to construct solutions to the JP vortex equations and what these solutions look like is given in \cite{Horvathy,HY,Horvathy:2008hd}. \\

In particular, it follows from \cite{Horvathy,HY} that, for a charge $2N$ Jackiw--Pi vortex, the map $f$ may, after a M\"obius transformation of the target, be normalised to the form
\begin{equation} \label{eq:normalised_Mobius}
f(z)=\frac{P(z)}{Q(z)},
\end{equation}
where $P$ and $Q$ are polynomials in $z$ such that $\deg P<\deg Q=N$. This gives a convenient way of constructing examples of JP vortices and thus vortex configurations on Nappi--Witten by lifting vortices from $M_{0}$.\\

\subsection{Vortex configurations}

Since $\NW$ is a trivial bundle over $M_0=\C$ with projection 
$\pi:\NW\to M_0$, vortex solutions on $\C$ can be lifted to 
geometric objects on $\NW$. 
We refer to these lifted objects as \emph{vortex configurations}.

The vortex equations on $\NW$ are the natural analogue of the
Jackiw--Pi vortex equations on $\C$, written in terms of the
left-invariant coframe $\sigma^1,\sigma^2$.

\begin{dfn}
A pair $(A,\Phi)$ consisting of a one-form $A$ on $\NW$
and a complex function $\Phi:\NW\to\C$
is called a vortex configuration if it satisfies
\begin{equation}
(\ud\Phi+\ui A\Phi)\wedge\sigma=0,
\qquad
F_A=-\frac{\ui}{2}|\Phi|^2\,\sigma\wedge\bar{\sigma},
\label{eq: 3D vortex equations}
\end{equation}
where $F_A=\ud A$ and $\sigma=\sigma^1+i\sigma^2$.
\end{dfn}

These equations are invariant under the abelian gauge transformations
\begin{equation}
(A,\Phi)\mapsto(A+\ud\alpha,e^{-i\alpha}\Phi),
\qquad \alpha\in C^\infty(\NW).
\end{equation}

In the three-dimensional case of \cite{Ross2021} and related work, these vortex configurations were related to a flat non-abelian connection on the group manifold, coming from a pullback of the Maurer-Cartan one-form. Much of that story carries over here. However, as it is not necessary for the construction of harmonic spinors we do not give the details here.

Vortex configurations on $\NW$ and vortices on $M_{0}$ are directly related through the following lemma.

\begin{lem}
\label{lem:lift_equivalence}
Let $\pi:\NW\to M_0$ and $s:M_0\to\NW$ be the projection and section defined in
Section~\ref{sec: Geometric conventions}.  
Then a pair $(a,\phi)$ on $M_0$ satisfies the Jackiw--Pi vortex equations
\eqref{eq: Vortex equations} (with $\lambda_0=0$, $\lambda=-1$) if and only if
the pair $(A,\Phi)=(-\pi^*a,\pi^*\phi)$ on $\NW$ satisfies the vortex configuration
equations \eqref{eq: 3D vortex equations}, up to an abelian gauge transformation.
Conversely, if $(A,\Phi)$ satisfies \eqref{eq: 3D vortex equations} and is basic
(i.e.\ pulled back from $M_0$), then $(-s^*A,s^*\Phi)$ satisfies \eqref{eq: Vortex equations}.
\end{lem}

\begin{proof}
The pullback of the base coframe is related to the left-invariant forms on $\NW$
by \eqref{eq:pullbackcoframe}.
Equivalently, in complex notation there exists a smooth phase $\chi(\theta)$ with
$|\chi|=1$ such that
\begin{equation}
\label{eq:complex_coframe_relation}
\pi^*e=\chi\,\sigma,
\qquad\text{where}\quad e=e^1+i e^2,\ \ \sigma=\sigma^1+i\sigma^2.
\end{equation}
In particular,
\begin{equation}
\label{eq:area_relation}
\pi^*(e\wedge\bar e)=\sigma\wedge\bar\sigma .
\end{equation}

Now set $(A,\Phi)=(-\pi^*a,\pi^*\phi)$. Then
\begin{equation}
\ud\Phi+\ui A\Phi
=\pi^*(\ud\phi-\ui a\,\phi).
\end{equation}
Using \eqref{eq:complex_coframe_relation},
\begin{equation}
(\ud\Phi+\ui A\Phi)\wedge\sigma
=\chi^{-1}\,\pi^*\bigl((\ud\phi-\ui a\phi)\wedge e\bigr),
\end{equation}
so the holomorphicity equation on $M_0$ is equivalent to the first vortex
configuration equation on $\NW$.

For the curvature equation, since $F_A=\ud A=-\pi^*(\ud a)=-\pi^*F_a$,
the Jackiw--Pi curvature equation
\begin{equation}
F_a=\frac{\ui}{2}|\phi|^2\,e\wedge\bar e
\end{equation}
pulls back to
\begin{equation}
F_A=-\frac{\ui}{2}|\Phi|^2\,\sigma\wedge\bar\sigma,
\end{equation}
where we used \eqref{eq:area_relation} and $|\Phi|=|\pi^*\phi|$.

Finally, the phase $\chi(\theta)$ in \eqref{eq:complex_coframe_relation}
can be absorbed by an abelian gauge transformation
$\Phi\mapsto e^{-i\alpha}\Phi$, $A\mapsto A+\ud\alpha$ (locally, and globally
on the trivial bundle), so the correspondence holds up to gauge.

The converse direction follows by pulling back along the section $s$, using
$s^*\pi^*=\mathrm{id}$, $s^*\sigma= \ui e$, and
the same identities.
\end{proof}

\begin{ex}
Consider the JP vortex given by $f(z)=z^{2}$. This representative is Möbius-equivalent, via the target inversion $w\mapsto 1/w$, to a normalized rational map of the form \eqref{eq:normalised_Mobius}, namely
$\tilde f(z)=1/z^2$.
The modulus of the Higgs field, following the general prescription in \cite{Manton:2016waw}, is
    \begin{equation}
        \vert\phi\vert^{2}=\frac{1}{\left(1+\vert f(z)\vert^{2}\right)^{2}}\left\vert\frac{\ud f}{\ud z}\right\vert^{2}=\frac{4\vert z\vert^{2}}{\left(1+\vert z\vert^{4}\right)^{2}},
    \end{equation}
    so we can take a gauge in which,
    \begin{equation}
        \phi=\frac{2z}{1+\vert z \vert^{4}},
    \end{equation}
    and away from the zeros of the Higgs field, $a_{\bar{z}}=-i\partial_{\bar{z}}\ln\phi$ so
    \begin{equation}
        a=\frac{2\ui \vert z\vert^{2}}{1+\vert z\vert^{4}}\left(z\ud\bar{z}-\bar{z}\ud z\right).
    \end{equation}
    The vortex configuration associated with this Jackiw--Pi vortex is thus
    \begin{equation}
       \Phi = \pi^{*}\phi =\frac{16\left(y-\ui x\right)}{16+(y^{2}+x^{2})^{2}}, \qquad A=-\pi^{*}a=4\frac{x^{2}+y^{2}}{16+\left(x^{2}+y^{2}\right)^{2}}\left[y\ud x-x\ud y\right].
       \label{eq: JP on NW example}
    \end{equation}
This example illustrates how vortex configurations arising from Jackiw--Pi vortices can be represented, up to gauge, on an $\R^2$-slice of $\NW$.
\end{ex}
\section{Harmonic spinors from vortices}
\label{sec: vortex harmonic spinors on NW}

We now use vortex configurations on $\NW$ to construct harmonic spinors for a twisted Dirac operator. Before stating the main result, we recall the chiral decomposition of the Dirac operator and introduce the notion of a vortex harmonic spinor.

The spinor bundle splits into left-handed and right-handed chiral parts. Thus, a Dirac spinor can be written as
\begin{equation}
    \Psi=
    \begin{pmatrix}
        \Psi_L\\
        \Psi_R
    \end{pmatrix},
\end{equation}
where $\Psi_L$ and $\Psi_R$ are two-component Weyl spinors.

Recall from equation~\eqref{eq:NW_dirac_operator} that, in the Weyl representation, the Dirac operator on $\NW$ is off-diagonal:
\begin{equation}
    D
    \begin{pmatrix}
        \Psi_L\\
        \Psi_R
    \end{pmatrix}
    =
    \begin{pmatrix}
        0 & D^-\\
        D^+ & 0
    \end{pmatrix}
    \begin{pmatrix}
        \Psi_L\\
        \Psi_R
    \end{pmatrix}
    =
    \begin{pmatrix}
        D^-\Psi_R\\
        D^+\Psi_L
    \end{pmatrix}.
\end{equation}
Therefore, the equation $D\Psi=0$ decouples into the two chiral equations
\begin{equation}
D^-\Psi_R=0,
\qquad
D^+\Psi_L=0.
\end{equation}
In what follows we focus on the right-handed equation.

\begin{dfn}
A \emph{vortex harmonic spinor} on $\NW$ is a pair $(\Psi_R,A)$ consisting of a right-handed Weyl spinor $\Psi_R:\NW\to\C^2$ and a one-form $A$ on $\NW$ such that
\begin{equation}
    D^-_{\NW,A}\Psi_R=0,
    \qquad
    F_A=-\frac{\ui}{2}\vert\Psi_R\vert^{2}\sigma\wedge\bar{\sigma}.
\end{equation}
\end{dfn}

We first record the equations implied by the vortex configuration equations.

\begin{lem}
\label{lem:contracted_vortex_eqs}
Let $(A,\Phi)$ be a vortex configuration on $\NW$, so that
\begin{equation}
    (\ud\Phi+\ui A\Phi)\wedge \sigma=0,
    \qquad
    F_A=-\frac{\ui}{2}|\Phi|^2\,\sigma\wedge\bar{\sigma}.
\end{equation}
Then
\begin{equation}
    U_0\Phi+\ui A_0\Phi=0,
    \qquad
    U_3\Phi+\ui A_3\Phi=0,
    \qquad
    U_+\Phi+\ui A_+\Phi=0,
    \label{eq:contracted_vortex_system}
\end{equation}
where $U_\pm=U_1\pm \ui U_2$ and $A_\mu=A(U_\mu)$.
\end{lem}

\begin{proof}
Set
\begin{equation}
\alpha:=\ud\Phi+\ui A\Phi.
\end{equation}
The first vortex configuration equation is $\alpha\wedge\sigma=0$. Contracting this $2$-form with the pair $(X,U_-)$ gives
\begin{equation}
(\alpha\wedge\sigma)(X,U_-)=\alpha(X)\sigma(U_-)-\alpha(U_-)\sigma(X).
\end{equation}
Now $\sigma=\sigma^1+\ui\sigma^2$, and since $\sigma(U_0)=\sigma(U_3)=\sigma(U_+)=0$ while $\sigma(U_-)\neq 0$, taking $X=U_0,U_3,U_+$ yields
\begin{equation}
\alpha(X)\sigma(U_-)=0.
\end{equation}
Hence $\alpha(U_0)=\alpha(U_3)=\alpha(U_+)=0$, which is precisely
\begin{equation}
U_0\Phi+\ui A_0\Phi=0,\qquad
U_3\Phi+\ui A_3\Phi=0,\qquad
U_+\Phi+\ui A_+\Phi=0.
\end{equation}
The curvature equation is just the second vortex configuration equation.
\end{proof}

We can now prove the main result of this section.

\begin{thm}
\label{thm: vort to spinor NW}
Let $(A,\Phi)$ be a vortex configuration on $\NW$. Then, for any constant $c \in \R$,
\begin{equation}
    \Psi_R=
    \begin{pmatrix}
        \Phi\\
        0
    \end{pmatrix},
    \qquad
    A_c=A+ c \sigma^3,
\end{equation}
defines a vortex harmonic spinor on $\NW$.
\end{thm}

\begin{proof}
Since $\sigma^3=\ud\theta$ is closed, and $c$ is a constant, 
\begin{equation}
F_{A_c}=\ud A_c=\ud A+c\,\ud\sigma^3=\ud A=F_A.
\end{equation}
Moreover,  $|\Psi_R|^2=|\Phi|^2$ and since $(A, \Phi)$ is a vortex configuration, 
\begin{equation}
    F_{A_c}=-\frac{\ui}{2}\vert\Psi_R\vert^{2}\sigma\wedge\bar{\sigma}.
\end{equation}

It remains to show that $D^-_{\NW,A_c}\Psi_R=0$. From the explicit form of the Dirac operator in Section~\ref{sec: NW and Dirac}, the chiral operator acting on right-handed Weyl spinors is
\begin{equation}
    D^-_{\NW,A_c}=
    \begin{pmatrix}
        U_{0}+\ui (A_c)_{0}+U_{3}+\ui (A_c)_{3}
        &
        U_{-}+\ui (A_c)_{-}
        \\[4pt]
        U_{+}+\ui (A_c)_{+}
        &
        U_{0}+\ui (A_c)_{0}-U_{3}-\ui (A_c)_{3}
        +\dfrac{3\ui}{2\sqrt{b+2}}
    \end{pmatrix}.
\end{equation}
Therefore
\begin{equation}
    D^-_{\NW,A_c}\Psi_R=
    \begin{pmatrix}
        (U_0+\ui (A_c)_0+U_3+\ui (A_c)_3)\Phi\\[4pt]
        (U_++\ui (A_c)_+)\Phi
    \end{pmatrix}.
\end{equation}

We first consider the lower component. Since $\sigma^3(U_+)=0$, we have
\begin{equation}
(A_c)_+=A_+,
\end{equation}
so
\begin{equation}
(U_++\ui (A_c)_+)\Phi=(U_++\ui A_+)\Phi=0
\end{equation}
by Lemma~\ref{lem:contracted_vortex_eqs}.

For the upper component, we use
\begin{align}
    (A_c)_0&=A_0+c \sigma^3(U_0)
    =A_0-\frac{c}{\sqrt{b+2}},\\
    (A_c)_3&=A_3+c\sigma^3(U_3)
    =A_3+\frac{c}{\sqrt{b+2}}.
\end{align}
Hence
\begin{equation}
    \ui (A_c)_0+\ui (A_c)_3 = \ui A_0+\ui A_3.
\end{equation}
Therefore
\begin{equation}
    (U_0+\ui (A_c)_0+U_3+\ui (A_c)_3)\Phi
    =
    (U_0+\ui A_0)\Phi+(U_3+\ui A_3)\Phi
    =0,
\end{equation}
again by Lemma~\ref{lem:contracted_vortex_eqs}.

Thus, both components vanish, and so
\begin{equation}
D^-_{\NW,A_c}\Psi_R=0.
\end{equation}
We conclude that $(\Psi_R,A_c)$ is a vortex harmonic spinor on $\NW$.
\end{proof}
Observe that the particular value $c=3/4$ is a natural normalisation, since it absorbs the constant spin-connection term in the lower-right entry of $D^-_{\NW,A_c}$. However, the zero-mode ansatz above only involves the first column of the operator, and hence the construction works for any constant $c$.
\section{Harmonic spinors on Minkowski}
\label{sec: vortex harmonic spinors on Minkowski}
It is well-known that the space of harmonic spinors is invariant under conformal transformations of the metric \cite{hitchin1974harmonic}. In fact, if two metrics are related by a conformal transformation $g_{\Omega}=\Omega^{2}g$ then the Dirac operators are related by 
\begin{equation}
    D_{\Omega}=\Omega^{-\frac{(n+1)}{2}}D\Omega^{\frac{n-1}{2}},
\end{equation}
where $n$ is the dimension of the manifold. 
In particular, if $\psi$ is harmonic with respect to $D$, then
\begin{equation}
\psi_\Omega=\Omega^{-\frac{n-1}{2}}\psi
\end{equation}
is harmonic with respect to $D_\Omega$.
In \cite{Erdos:2001} this is discussed in detail for the conformal relationship between $S^{3}$ and $\R^{3}$, and the relationship described there underlies the relationship between the vortex magnetic modes of \cite{RS2,RS1}.

\subsection{Nappi--Witten is conformally flat}
The bi-invariant metrics on $\mathcal{N}$ induced by the metrics \eqref{eq: NW quadratic form} at the identity (taking $k=1$) are 

\begin{equation}\label{metric_NW}
ds^2= dx^2 + dy^2 + 2 d\theta d t+(y dx - x dy )d \theta + b d \theta^2.
\end{equation}
Under the change of variables
\begin{align}
        x' &= \cos (\theta/2) x + \sin (\theta/2) y \\
    y' &= - \sin (\theta/2) x + \cos (\theta/2) y \\
    t' &= t +\frac{b}{2} \theta, \label{eq:t_translate}
\end{align}
one sees that the metric is that of a plane pp-wave, in particular, it is a Cahen-Wallach space \cite{leistner2022conformal},
\begin{equation}
    ds^2 = dx'^2 + dy'^2 + 2 d\theta d t'   -  \frac{x'^2 + y'^2}{4} d \theta^2.
\end{equation}
For $\theta \neq \pi + 2k\pi$, the change of variables,
\begin{align}
    \vec{X} &= \frac{1}{\cos(\theta/2)} (x',y')\\
    U &= 2 \tan(\theta/2) \\
    V &= -t' +\frac{\left((x')^{2}+(y')^{2}\right)}{4}\tan\left(\theta/2\right),
\end{align}

gives
\begin{equation}
    d\vec{X}^2 - 2 dU dV = \frac{1}{\cos^2(\theta/2)}  ds^2.
    \label{eq:confromal relation between metrics}
\end{equation}
In other words, the bi-invariant metrics on $\mathcal{N}$ are conformal to Minkowski space $\mathcal{M}$, for which we use light-cone-type coordinates $U, V$ for a two-dimensional Lorentzian subspace, and   $\vec X=(X^1,X^2)$ for the transverse Euclidean two-plane.

More precisely, let $
\mathcal{N}^* := \mathcal{N} \setminus \{\theta = \pi + 2k\pi \mid k \in \mathbb{Z}\}$, then the inverse of the above change of variables is the map $H: \mathcal{M} \rightarrow \mathcal{N}^*$ is given by,
\begin{align}
\theta		&=2\arctan\frac{U}{2},\\
\begin{pmatrix}
    x\\
    y
\end{pmatrix}		&=\frac{2}{4+U^{2}}\begin{pmatrix}
2&-U\\
U&~2
\end{pmatrix}\vec{X},\\
t&=-V-b\arctan\frac{U}{2}+\frac{1}{2}\frac{U}{4+U^{2}}\vec{X}^{2}.
\end{align}

From \eqref{eq:confromal relation between metrics} we obtain the conformal relation
\begin{equation} \label{conformal_relation}
    ds^2_{\mathcal{M}} = \Omega^2\, H^* ds^2,
    \qquad 
    \Omega=\frac{1}{\cos\left(\theta/2\right)} = \frac{\sqrt{4+U^2}}{2}.
\end{equation}
It follows that the Dirac operators on Minkowski space and Nappi--Witten are related by
\begin{equation}
    D_{\mathcal{M}} = \Omega^{-\frac{5}{2}}\, H^* D_{\NW}\, \Omega^{\frac{3}{2}}.
\end{equation}

\subsection{Spinors from vortices}

From here on we are working with $b=0$.
Next, we investigate the relationship between the orthonormal coframe in Minkowski space
$
e \;=\; (e^1,e^2,e^3,e^0)^\top,
$
with
\begin{equation}
    e^1=dX^1,\quad e^2=dX^2,\quad
e^3=\frac{dU-dV}{\sqrt2},\quad
e^0=\frac{dU+dV}{\sqrt2}\,.
\end{equation}
and the pullback via $H$ of the orthonormal coframe in Nappi--Witten, 
\begin{equation}
\vartheta^1=\sigma^1,\qquad
\vartheta^2=\sigma^2,\qquad
\vartheta^3=\frac{1}{\sqrt2}(\sigma^3+\sigma^4),\qquad
\vartheta^0=\frac{1}{\sqrt2}(\sigma^3-\sigma^4).
\label{eq:NW orthonormal coframe}
\end{equation}

Since the metrics in Minkowski space satisfy \eqref{conformal_relation},
the rescaled coframe
\begin{equation}
\widehat{\vartheta}^\mu:=\Omega\,H^*\vartheta^\mu
\end{equation}
is orthonormal for the Minkowski metric.

\begin{lem}
\label{lem:Lambda}
The proper orthochronous Lorentz transformation 
$\Lambda(U,\vec X)\in SO^+(1,3)$ relating the coframes $e$ and $(\widehat{\vartheta}^1,\widehat{\vartheta}^2,\widehat{\vartheta}^3,\widehat{\vartheta}^0)^\top$ as
\begin{equation}
e^a=\Lambda^a_{\;\;\mu}(U,\vec X)\,\widehat{\vartheta}^\mu,
\label{eq:Lambda frame relation}
\end{equation}
is,
\begin{equation}
\Lambda(U,\vec X)= \frac{1}{\Omega}
\begin{pmatrix}
1 & -\dfrac{U}{2} &
\dfrac{UX^1+2X^2}{4\sqrt2} &
\dfrac{UX^1+2X^2}{4\sqrt2}
\\[10pt]
\dfrac{U}{2} & 1 &
-\dfrac{2X^1-UX^2}{4\sqrt2} &
-\dfrac{2X^1-UX^2}{4\sqrt2}
\\[10pt]
-\dfrac{X^2}{2\sqrt2} & \dfrac{X^1}{2\sqrt2} &
\dfrac{16+2U^2-r^2}{16} &
-\dfrac{r^2-2U^2}{16}
\\[10pt]
\dfrac{X^2}{2\sqrt2} & -\dfrac{X^1}{2\sqrt2} &
\dfrac{r^2+2U^2}{16} &
\dfrac{16+2U^2+r^2}{16}
\end{pmatrix},
\label{eq:Lambda matrix}
\end{equation}
with $r^2=(X^1)^2+(X^2)^2$.
\end{lem}

\begin{proof}
The pullback by $H$ of the left-invariant one-forms defined in Section~\ref{ssec:Group_conventions} is,
\begin{align}
H^*\!\begin{pmatrix}\sigma^1\\ \sigma^2\end{pmatrix}
&= \frac{2}{4+U^2}\!\left( \mathcal{C}(U) \,d\vec X + J\,\vec X\,dU\right),\\
H^*\sigma^3 &= \frac{4}{4+U^2}\,dU,\\
H^*\sigma^4 &= -\,dV - \frac{\lvert\vec X\rvert^2}{2(4+U^2)}\,dU
+ \frac{U}{4+U^2}\,\vec X\!\cdot\! d\vec X
- \frac{2}{4+U^2}\,\bigl(X^1\,dX^2 - X^2\,dX^1\bigr),
\end{align}
where
\begin{equation}
\mathcal{C}(U)=\begin{pmatrix}2&U\\ -U&2\end{pmatrix},\qquad
J=\begin{pmatrix}0&-1\\ 1&0\end{pmatrix}.
\end{equation}

From this computation, we obtain that the flat Minkowski coframe $e$ is related to the pullback of the Nappi--Witten coframe $(\sigma^i)$ by
\begin{equation}
e^a = -\frac{1}{\Omega}\, G^a_{\;\;i}(U,\vec{X})\, H^*\sigma^i,
\label{eq: frame relations}
\end{equation}
where $G(U,\vec{X})$ is the transformation
\begin{equation}
G(U,\vec X)=
\begin{pmatrix}
-\Omega & \dfrac{U\Omega}{2} & -\dfrac{\kappa(UX^1 + 2X^2)}{D} & 0\\[12pt]
-\dfrac{U\Omega}{2} & -\Omega & \dfrac{\kappa(2X^1 - UX^2)}{D} & 0\\[12pt]
\dfrac{X^2\Omega}{2\sqrt{2}} & -\dfrac{X^1\Omega}{2\sqrt{2}} & \dfrac{\kappa}{\sqrt{2}}\!\left(\dfrac{r^2}{2D}-1\right) & -\dfrac{\Omega}{\sqrt{2}}\\[12pt]
-\dfrac{X^2\Omega}{2\sqrt{2}} & \dfrac{X^1\Omega}{2\sqrt{2}} & -\dfrac{\kappa}{\sqrt{2}}\!\left(\dfrac{r^2}{2D}+1\right) & \dfrac{\Omega}{\sqrt{2}}
\end{pmatrix},
\end{equation}
with
\begin{equation}
D = 4+U^2,\qquad
\Omega(U) = \frac{\sqrt{D}}{2},\qquad
\kappa(U) = \frac{D^{3/2}}{8},\qquad
r^2 =(X^1)^2+(X^2)^2.
\end{equation}

Using \eqref{eq:NW orthonormal coframe}, the left-invariant one-forms are expressed in terms of the orthonormal coframe as
\begin{equation}
\begin{pmatrix}
\sigma^1\\
\sigma^2\\
\sigma^3\\
\sigma^4
\end{pmatrix}
=
M
\begin{pmatrix}
\vartheta^1\\
\vartheta^2\\
\vartheta^3\\
\vartheta^0
\end{pmatrix},
\qquad
M=
\begin{pmatrix}
1&0&0&0\\
0&1&0&0\\
0&0&\frac1{\sqrt2}&\frac1{\sqrt2}\\
0&0&\frac1{\sqrt2}&-\frac1{\sqrt2}
\end{pmatrix}.
\end{equation}

From \eqref{eq: frame relations} 
and the rescaling $\widehat{\vartheta}^\mu=\Omega H^*\vartheta^\mu$, we obtain
\begin{equation}
e^a
=
-\frac{1}{\Omega}\,G^a_{\;\;i}\,M^i_{\;\;\mu}\,H^*\vartheta^\mu =
-\frac{1}{\Omega^2}\,G^a_{\;\;i}\,M^i_{\;\;\mu}\,\widehat{\vartheta}^\mu.
\end{equation}
Thus
\begin{equation}
\Lambda=-\frac{1}{\Omega^2}GM.
\end{equation}

Substituting the explicit expression for $G$ and simplifying using
$\frac{\kappa}{D}=\frac{\Omega}{4}$,
yields the matrix \eqref{eq:Lambda matrix}.

\end{proof}

\medskip

\begin{cor}
\label{cor: NW to Mink}
Let $\Psi:\NW\to\C^4$ be a harmonic spinor for the Dirac operator $D_{\NW,A}$ on $\NW$ coupled to a $U(1)$ connection $A$. Let $S(\Lambda)$ denote a spin lift of the Lorentz transformation $\Lambda$, satisfying
\begin{equation}
S(\Lambda)^{-1}\gamma^a S(\Lambda)=\Lambda^a_{\;\;b}\gamma^b.
\end{equation}
Then
\begin{equation}
\Psi_{\mathcal M}=S(\Lambda)\,\Omega^{-\frac32}\,H^*\Psi
\end{equation}
is a harmonic spinor for the Dirac operator $D_{\mathcal{M}}$ on Minkowski space coupled to the pulled-back connection $H^*A$.
\end{cor}

To write down a spin lift of the Lorentz transformation \(\Lambda(U,\vec X)\), it is convenient to use the Weyl basis adapted to the ordering \((1,2,3,0)\). Define
\begin{equation}
\tau^1=\sigma_1,\qquad
\tau^2=\sigma_2,\qquad
\tau^3=\sigma_3,\qquad
\tau^0=I_2,
\end{equation}
and
\begin{equation}
\bar\tau^1=-\sigma_1,\qquad
\bar\tau^2=-\sigma_2,\qquad
\bar\tau^3=-\sigma_3,\qquad
\bar\tau^0=I_2.
\end{equation}
Then the gamma matrices are
\begin{equation}
\Gamma^a=
\begin{pmatrix}
0 & \tau^a\\
\bar\tau^a & 0
\end{pmatrix},
\qquad a\in\{1,2,3,0\}.
\end{equation}

Introduce
\begin{align}
  \beta(U,\vec X)&:=\frac{X^2+\ui X^1}{2(4+U^2)^{1/4}},\\  
  \chi(U)
&:=
\exp\!\left(\frac{\ui}{2}\arctan\frac{U}{2}\right)
=
\frac{\sqrt{4+U^2}+2+\ui U}
{\sqrt{2\sqrt{4+U^2}\,(\sqrt{4+U^2}+2)}}.
\end{align}

A direct computation shows that the matrix
\begin{equation}
B(\Lambda)=
\begin{pmatrix}
\sqrt{\Omega}\,\chi & \beta\,\chi^{-1}\\[4pt]
0 & \Omega^{-1/2}\,\chi^{-1}
\end{pmatrix}
\label{eq:A-spin-lift}
\end{equation}
belongs to \(SL(2,\C)\) and satisfies
\begin{equation}
B(\Lambda)\,\tau^a\,B(\Lambda)^\dagger
=
\Lambda^a_{\;\;b}\,\tau^b.
\end{equation}

Hence a spin lift of \(\Lambda\) is given by
\begin{equation}
S(\Lambda)=
\begin{pmatrix}
B(\Lambda)^{-1} & 0\\[4pt]
0 & B(\Lambda)^\dagger
\end{pmatrix} = 
\begin{pmatrix}
\Omega^{-1/2}\chi^{-1} & -\beta\,\chi^{-1} & 0 & 0\\[4pt]
0 & \sqrt{\Omega}\,\chi & 0 & 0\\[4pt]
0 & 0 & \sqrt{\Omega}\,\chi^{-1} & 0\\[4pt]
0 & 0 & \bar\beta\,\chi & \Omega^{-1/2}\chi
\end{pmatrix}.
\label{eq:explicit-spin-lift}
\end{equation}

\subsection{Examples}

\begin{ex}
Consider the example from Section~\ref{sec: vortices}, where a Jackiw--Pi vortex on $\C$ with $f(z) = z^2$ is lifted to $\NW$ and the fields are given by \eqref{eq: JP on NW example} as
\begin{equation}
\Phi = \pi^{*}\phi =\frac{16\left(y-\ui x\right)}{16+(y^{2}+x^{2})^{2}}, \qquad
A=-\pi^{*}a=4\frac{x^{2}+y^{2}}{16+\left(x^{2}+y^{2}\right)^{2}}\left[y\ud x-x\ud y\right].
\end{equation}
From Theorem~\ref{thm: vort to spinor NW}, the corresponding right-handed Weyl spinor on $\NW$ is
\begin{equation}
\Psi_R =
\begin{pmatrix}
\frac{16\left(y-\ui x\right)}{16+(y^{2}+x^{2})^{2}} \\[6pt]
0
\end{pmatrix}.
\end{equation}
Using the conformal map $H:\mathcal{M}\to\mathcal{N}^*$, we have
\begin{equation}
y - \ui x = \frac{2}{4 + U^2} \left[ (U - 2\ui) X^1 + (2 + \ui U) X^2 \right], \quad x^2 + y^2 = \frac{4 r^2}{4 + U^2},
\end{equation}
with $r^2 = (X^1)^2 + (X^2)^2$. This gives
\begin{equation}
H^* \Phi = \frac{ 2 \left[ (U - 2 \ui ) X^1 + (2 + \ui U) X^2 \right] (4 + U^2) }{ (4+U^2)^2 + r^4 }.
\end{equation}
The conformally rescaled Dirac spinor is
\begin{equation}
\widetilde{\Psi}
=
\Omega^{-\frac{3}{2}}
\begin{pmatrix}
0\\
0\\
H^*\Phi \\
0
\end{pmatrix},
\end{equation}

Applying the spin lift $\Psi_{\mathcal{M}} = S(\Lambda)\,\widetilde{\Psi}$, the norm becomes
\begin{equation}
\vert\Psi_{\mathcal{M}}\vert^{2}
= \left(\frac{2(4+U^2) + r^2}{4\sqrt{4+U^2}}\right)\Omega^{-3} \vert H^*\Phi \vert^2 =\frac{8r^{2}(4+U^2)\left[2(4+U^2)+r^2\right]}{\left((4+U^{2})^2+r^4\right)^{2}}.
\label{eq: Mink spinor norm}
\end{equation}
The norm of the spinor in terms of the coordinates $U,r=\sqrt{(X^{1})^{2}+(X^{2})^{2}}$ is shown in figure \ref{fig: spinor norm plot}.

\begin{figure}
    \centering
    \includegraphics[width=0.7\linewidth]{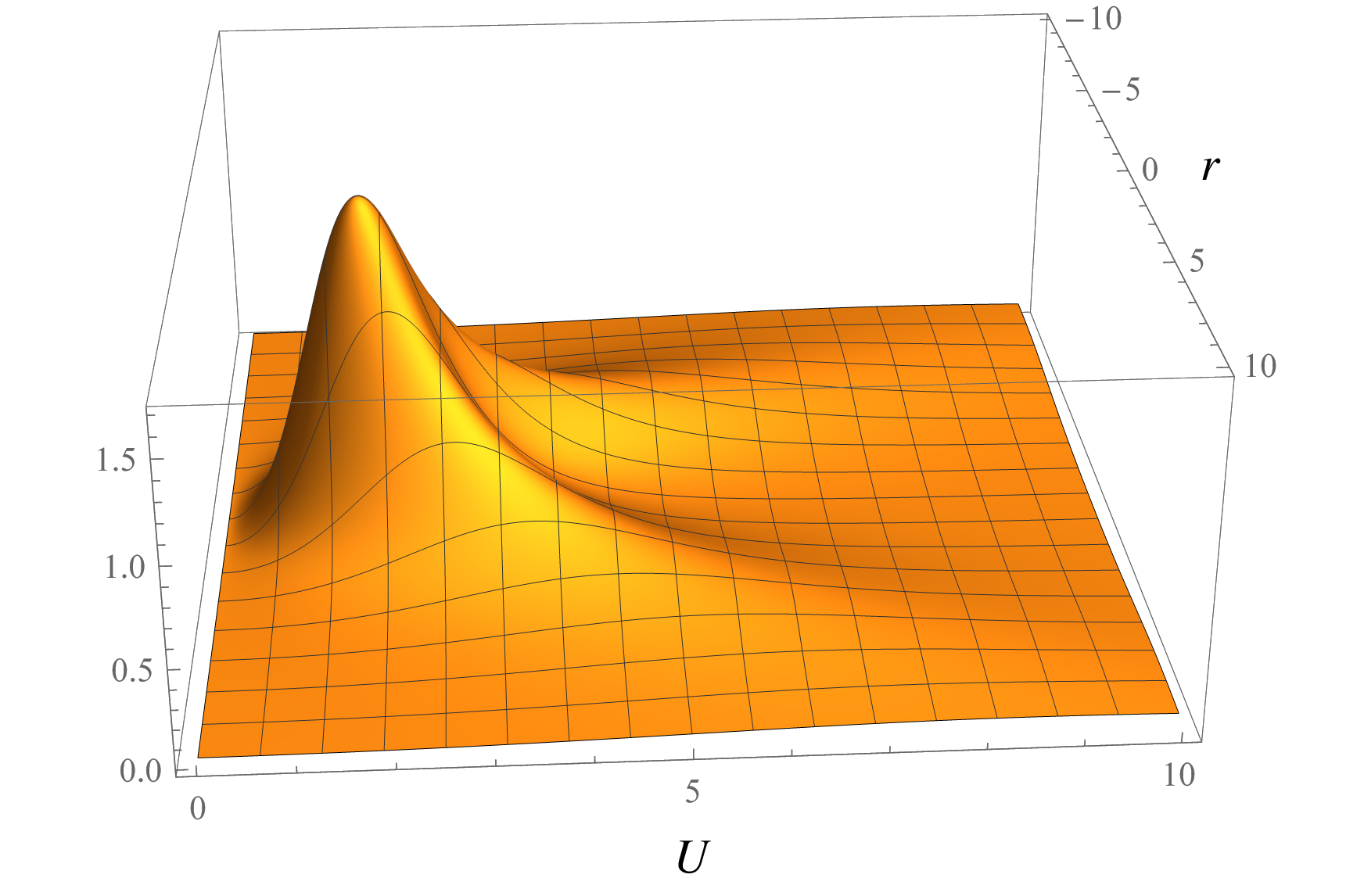}
    \caption{A plot of the norm of the spinor $\vert\Psi_{\mathcal{M}}\vert^{2}$ from \eqref{eq: Mink spinor norm} in terms of the Minkowski coordinates $U,\, r=\sqrt{(X^{1})^{2}+(X^{2})^{2}}$. }
    \label{fig: spinor norm plot}
\end{figure}
\end{ex}

\begin{ex}
More generally, consider the Jackiw--Pi vortex on $\C$ with $f(z) = z^m$ for topological charge $m$:
\begin{equation}
\phi = \frac{m z^{m-1}}{1 + |z|^{2m}},
\qquad
a = \frac{m \ui |z|^{2m-2}}{1 + |z|^{2m}} (z \ud \bar{z} - \bar{z} \ud z).
\end{equation}
Lifted to $\NW$, the vortex configuration is:
\begin{equation}
\Phi = \pi^* \phi = \frac{m2^{m+1} (y - \ui x)^{m-1}}{2^{2m} + (x^2 + y^2)^m},
\qquad
A = -\pi^* a = \frac{2m(x^2 + y^2)^{m-1}}{2^{2m} + (x^2 + y^2)^m} [y \ud x - x \ud y].
\end{equation}
The corresponding right-handed Weyl harmonic spinor on $\NW$ is
\begin{equation}
\Psi_R =
\begin{pmatrix}
\frac{m2^{m+1} (y - \ui x)^{m-1}}{2^{2m} + (x^2 + y^2)^m} \\[6pt]
0
\end{pmatrix},
\end{equation}
\begin{equation}
A_c = A + c \sigma^3.
\end{equation}
Pulling back to Minkowski space gives
\begin{equation}
H^* \Phi =
\frac{m (4 + U^2) \left[ (U - 2\ui)X^1 + (2 + \ui U)X^2 \right]^{m-1}}
{(4+U^2)^m + r^{2m}},
\end{equation}
\begin{equation}
r^2 = (X^1)^2 + (X^2)^2.
\end{equation}
The conformally rescaled Dirac spinor is
\begin{equation}
\widetilde{\Psi}
=
\Omega^{-\frac{3}{2}}
\begin{pmatrix}
0\\
0\\
H^*\Phi \\
0
\end{pmatrix}.
\end{equation}

Applying the spin lift $\Psi_{\mathcal{M}} = S(\Lambda)\,\widetilde{\Psi}$, we obtain
\begin{equation}
\vert\Psi_{\mathcal{M}}\vert^{2}
= \left(\frac{2(4+U^2) + r^2}{4\sqrt{4+U^2}}\right)\Omega^{-3} \vert H^*\Phi \vert^2.
\end{equation}
Using
\begin{equation}
\left| (U - 2\ui)X^1 + (2 + \ui U)X^2 \right|^2 = (4+U^2)r^2,
\end{equation}
a direct computation yields
\begin{equation}
\vert\Psi_{\mathcal{M}}\vert^{2}
=
\frac{2m^2 r^{2m-2} (4+U^2)^{m-1} \left[2(4+U^2)+r^2\right]}
{\left((4+U^{2})^m+r^{2m}\right)^{2}}.
\end{equation}
Setting $m=2$ recovers \eqref{eq: Mink spinor norm}.
\end{ex}

\section{Summary and Outlook}
\label{Sec: summary}

This paper provides a construction of harmonic spinors on Minkowski space in terms of Jackiw--Pi vortices on $\R^{2}$. The approach developed here can be viewed as a four-dimensional extension of the framework introduced in \cite{Ross2021,RS2,RS1}, and completes the geometric correspondence between integrable vortex equations \cite{Manton:2016waw} and harmonic spinors.

From a physical perspective, our results give an explicit construction of spinor fields on $\R^{1,3}$ solving the massless Dirac equation in the presence of a background Abelian magnetic field. The intermediate construction on the Nappi--Witten space yields, to our knowledge, the first explicit examples of harmonic spinors on $\NW$.

In the spherical case \cite{RS1}, vortex harmonic spinors were related to ultra-cold atom systems involving linked magnetic fields. In the hyperbolic case \cite{RS2}, they were used in the construction of Yang--Mills fields from data on anti-de Sitter space \cite{Hirpara:2023htg,Kumar:2023sia}. It would be interesting to investigate whether the four-dimensional spinors constructed here admit a similar physical interpretation or application.


\section*{Acknowledgments}

CR would like to thank Bernd Schroers for many useful discussions about magnetic zero-modes and vortices, and Lennart Schmidt for some useful discussions at a very early stage of this project. RSG would like to thank Ivo Terek for many helpful discussions on pp-waves.

\bibliographystyle{abbrv}
\bibliography{nw_harmonic_spinors}
\end{document}